\documentclass[a4paper,runningheads]{llncs}
\usepackage[T1]{fontenc}
\usepackage{graphicx}

\usepackage{amssymb,amsmath}%
\usepackage{mathtools}
\usepackage{thmtools}
\usepackage[labelfont={bf,sf}]{subcaption}
\usepackage[misc,geometry]{ifsym}

\let\oldendproof\endproof
\renewcommand\endproof{~\hfill$\qed$\oldendproof}

\usepackage{apptools}
\newcommand{\restateref}[1]{\IfAppendix{\hyperref[#1]{$\star$}}{\hyperref[#1*]{$\star$}}}

\usepackage{todonotes}

\usepackage{setspace}
\usepackage{hyperref}
\usepackage{cleveref}
\usepackage[bibliography=common]{apxproof}

\crefname{section}{Sec.}{Secs.}
\Crefname{section}{Section}{Sections}
\crefname{figure}{Fig.}{Figs.}
\crefname{table}{Tab.}{Tabs.}
\crefname{Definition}{Def.}{Defs.}
\crefname{theorem}{Theorem}{Theorems}
\crefname{corollary}{Corollary}{Corollaries}
\crefname{claim}{Claim}{Claims}
\crefname{myclaim}{Claim}{Claims}
\crefname{lemma}{Lemma}{Lemmas}
\crefname{proposition}{Claim}{Claims}
\crefname{equation}{Eq.}{Eqs.}
\crefname{appendix}{App.}{Apps.}

\usepackage{color}

\newtheorem{myclaim}[proposition]{Claim}

\newcommand{\oh}{\ensuremath{\mathcal{O}}}

\DeclareMathOperator{\cro}{\chi}
\DeclareMathOperator{\natpos}{p_\mathsf{nat}}
\DeclareMathOperator{\optpos}{P_\mathsf{opt}}

\graphicspath{{figures/}}

\begin{document}

\title{Simultaneous Drawing of Layered Trees}
\author{Julia~Katheder \inst{1}\orcidID{0000-0002-7545-0730} \and
Stephen~G.~Kobourov\inst{2}\orcidID{0000-0002-0477-2724} \and \\
Axel~Kuckuk\inst{1}\orcidID{0000-0002-5070-3412} \and \\
Maximilian~Pfister\inst{1}\orcidID{0000-0002-7203-0669} \and \\
Johannes~Zink\inst{3}\orcidID{0000-0002-7398-718X}}
\authorrunning{J.~Katheder, S.G.~Kobourov, A.~Kuckuk, M.~Pfister, and J.~Zink}
\institute{Wilhelm-Schickard-Institut f{\"u}r Informatik, Universit{\"a}t T{\"u}bingen, T{\"u}bingen, Germany \\
\email{firstname.lastname@uni-tuebingen.de} \and
Department of Computer Science, University of Arizona, Tucson, USA
\email{lastname@cs.arizona.edu} \and
Institut f{\"u}r Informatik, Universit{\"a}t W{\"u}rzburg, W{\"u}rzburg, Germany
\email{lastname@informatik.uni-wuerzburg.de}}
\maketitle
\begin{abstract}
	We study the crossing-minimization problem in a layered graph drawing
	of planar-embedded rooted trees whose leaves have a given total order on the first layer,
	which adheres to the embedding of each individual tree.
	The task is then to permute the vertices
	on the other layers
	(respecting the given tree embeddings)
	in order to minimize the number of crossings.
	While this problem is known to be NP-hard for multiple trees even on just two layers,
	we describe a dynamic program running in polynomial time
	for the restricted case of two trees.
	If there are more than two trees, we restrict the number of layers to three,
	which allows for a reduction to a shortest-path problem.
	This way, we achieve XP-time in the number of trees.
	
\end{abstract}

\section{Introduction}
Visualizing hierarchical structures as directed trees is essential for many applications, from software engineering~\cite{chawathe1996change} to medical ontologies~\cite{bodenreider2004unified} and phylogenetics in biology~\cite{treeVizReferenceGuide}.
Phylogenetic trees in particular can serve as an example to illustrate the challenges of working with hierarchical structures, as they are inferred from large amounts of data using various computational methods \cite{yang2012molecular} and need to be analyzed and checked for plausibility using domain knowledge~\cite{liu2019aggregated}.
From a human perspective, visual representations are needed for this purpose. Most available techniques focus on the visualization of a single tree~\cite{graham2010survey}.
However, certain tasks may require working with multiple, possible interrelated trees, such as the comparison of trees~\cite{liu2019aggregated,munzner2003treejuxtaposer} or analyzing ambiguous lineages~\cite{puigbo2009search}.
Graham and Kennedy~\cite{graham2010survey} provide a survey for drawing multiple trees in this context.

While there are many different visualization styles for trees (see an overview by Schulz~\cite{schulz2011treevis}), directed node-link diagrams are the standard.
The most common approach to visualize a directed graph as a node-link diagram is the layered drawing approach due to Sugiyama et al.~\cite{4308636}.
After assigning vertices to layers, the next step is to permute the vertices on each layer
such that the number of crossings is minimized, as crossings negatively affect the readability of a graph drawing~\cite{Purchase2002,Ware2002}.
However, this problem turns out to be hard even when restricting the number of layers or the type of graphs.
For example, if the number of layers is restricted to two, crossing minimization remains NP-hard for general graphs~\cite{NPHardGeneral}, even if the permutation on one layer is fixed~\cite{FixedLayerNPHardGeneral}, known as the \emph{one sided crossing minimization} (OSCM) problem.
However, it is known that OSCM is fixed-parameter tractable in the number of crossings, which has first been shown by Dujmovic and Whitesides~\cite{DBLP:journals/algorithmica/DujmovicW04}.
For the special case of a single tree on two layers, OSCM can be solved in polynomial time~\cite{kLevel} and in the case that both layers are variable, the problem can be reduced to the minimum linear arrangement problem~\cite{crossEqualLinArrangement},
which is polynomial-time solvable~\cite{CHUNG198443}.
For an arbitrary number of layers, the problem is still NP-hard even for trees~\cite{kLevel},
however, the obtained trees in the reduction~\cite{kLevel} are not drawn upward in the direction from the leaves to a root vertex (and we do not see an obvious way to adjust their construction).
With respect to forests, the general case where $k \in \oh(n)$ is known to be NP-hard~\cite{DBLP:conf/gd/MunozUV01} even for $\ell = 2$ layers
and trees of maximum degree $4$.

\paragraph{Our Contributions.}
We consider the crossing-minimization problem
for an $n$-vertex forest of $k$ trees whose vertices are assigned to $\ell$ layers
such that all leaves are on the first layer in a fixed total order
and the vertices on each of the other layers need to be permuted.
In other words, the task is to draw $k$ layered rooted trees
whose leaves may interleave simultaneously,
while minimizing the number of crossings.

We show that the case of $k = 2$ trees is polynomial-time solvable
for arbitrary $\ell$ using a dynamic program (see \cref{sec:two-trees}).
Furthermore, we describe an XP-algorithm\footnote{%
	XP is a parameterized running-time class and
	an XP-algorithm has a running time in $\oh(|I|^{f(k)})$,
	where $|I|$ is the size of the instance,
	$f$ a computable function, and $k$ the parameter.
	Note that every FPT-algorithm is an XP-algorithm
	but not vice versa.}
in the number~$k$ of trees
modeling the solution space by a $k$-dimensional grid graph
for $\ell = 3$ layers.
Our result generalizes to planar graphs under certain conditions (see \cref{sec:three-layers}).
We conclude with the
open case of $k \geq 3$ and $\ell \geq 4$ (see \cref{sec:open}).

\section{Preliminaries}
\label{sec:prel}

\begin{figure}[t]
	\begin{subfigure}[b]{.645\textwidth}
			\centering
			\includegraphics[scale=1,page=1]{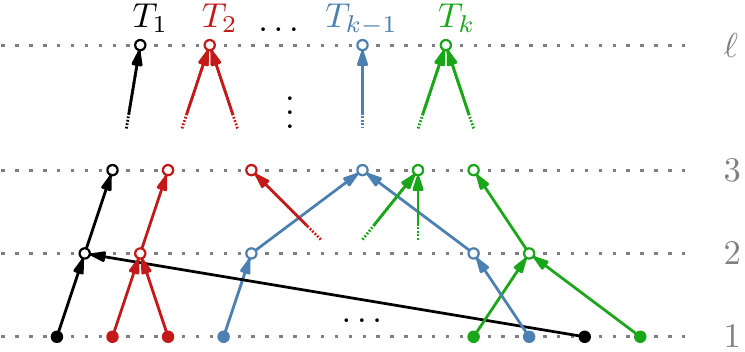}
			\subcaption{}
			\label{fig:sim-embedding-trees-a}   
	\end{subfigure}
	\begin{subfigure}[b]{.345\textwidth}
			\centering
			\includegraphics[scale=1,page=2]{sim-embedding-trees}
			\subcaption{}
			\label{fig:sim-embedding-trees-b}   
	\end{subfigure}
		\caption{\textsf{\bfseries (a)} Upward drawing of $k$ disjoint directed rooted trees $T_1,\dots T_k$ on $\ell$ layers.
			As indicated by the filled vertices, the total order $<_1$ of layer~$1$ is given,
			while the total orders $<_2,\dots, <_\ell$ need to be determined.
			In the following figures, we drop the arrowheads and assume an upward direction.
			\textsf{\bfseries (b)} Illustration of positions (gray boxes) with respect to
			$T_1$ and their respective \emph{ideal positions} indicated by a directed gray arrow from each position $p$ to its ideal position~$p^\star$.}
		\label{fig:sim-embedding-trees}
	\end{figure}

Let $\mathcal{F}$ be a given
forest of $k$ disjoint rooted trees $T_1,\dots, T_k$
directed towards the roots such that all vertices
except for the roots have out-degree~1.
For an integer $\ell \ge 2$,
let an assignment of the vertices to $\ell$ layers be given, 
such that each tree $T_i$ is \emph{drawn upward}, i.e.,
for any directed edge $(u,v) \in T_i$, we have that the layer of~$u$, denoted by~$L(u)$, is strictly less than $L(v)$.
This implies that if $L(u) = 1$, $u$ is a leaf of $T_i$.
The other way around, we also require that
for any leaf~$v$, $L(v) = 1$. Note that the roots of the trees can be placed on different layers, while layer~$\ell$ hosts the root of every tree with height exactly $\ell$.
We refer to the set of vertices of~$T_i$ on layer~$j$ as $V_j(T_i)$ and we define the set of all vertices on layer~$j$ as $V_j(\mathcal{F}) = V_j(T_1) \cup \dots \cup V_j(T_k)$.

We further require that the total order $<_1$ of layer~$1$
(i.e., the order of all leaves)
is given as part of the input, with the additional restriction that $<_1$ induces a planar embedding~$\mathcal{E}_i$
with respect to each individual tree $T_i$, that is, there exists an ordering of the (non-leaf) vertices of~$T_i$ such that no two edges of $T_i$ cross, see \cref{fig:sim-embedding-trees-a} for an illustration. Since the leaves of each $T_i$ are all on layer~$1$, the embedding~$\mathcal{E}_i$ is unique and implies a total ordering~$<_j^i$ of the vertices of $T_i$
on every layer $j \in \{2, \dots, \ell\}$.
Therefore, we henceforth assume that $V_j(T_i)$ appears in the corresponding
vertex order $<_j^i$,
and if we combine all $<_j^i$ for $i \in \{1, \dots, k\}$,
we obtain a partial order, which we call~$\prec_j$.

The task is to find a total order $<_j$ of $V_j(\mathcal{F})$
extending the partial order $\prec_j$
for each $j \in \{2,\dots,\ell\}$
such that the total number of pairwise edge crossings implied by a corresponding straight-line realization of $\mathcal{F}$ is minimized.

We restrict the notion of an upward drawing even further since we require that for any directed edge $(u,v) \in T_i$, we have that $L(u) + 1 = L(v)$.
If our input does not fulfill this requirement, this can be achieved by subdividing edges which span several layers (as commonly done, e.g., in the Sugiyama framework).
In the following, we assume that $n$ is the number of vertices after subdivision
and let $n_1, \dots, n_k$ be the number of vertices of $T_1, \dots, T_k$, respectively.
Furthermore, we denote the number of vertices of tree~$T_i$ on layer~$j$ by $n_{i|j} = |V_j(T_i)|$.

\section{Two Trees on Arbitrarily Many Layers}
\label{sec:two-trees}

In this section, we assume that we are given a forest
$\mathcal{F} = \{T_1, T_2\}$ with embeddings~$\mathcal{E}_1$ and~$\mathcal{E}_2$.
We fix the drawing of $T_1$ according to~$\mathcal{E}_1$
and the only remaining task is to add the non-leaf vertices of~$T_2$
in the order prescribed by~$\mathcal{E}_2$
such that the number of crossings is minimized.
To this end, we describe a dynamic programming approach,
which leads to the following theorem.

\begin{theorem}
	\label{thm:Two-trees}
	Let $\mathcal{F}$ be an $n$-vertex layered forest of two rooted trees,
	where all leaves are assigned to layer~$1$ and have a fixed order,
	which prescribe a planar embedding of each tree individually.
	We can compute a drawing of $\mathcal{F}$
	where each tree is drawn in the prescribed planar embedding
	with the minimum number of crossings in $\oh(n^3)$ time.
\end{theorem}

\begin{proof}
	As stated before, we fix the drawing of $T_1$
	according to~$\mathcal{E}_1$.
	Hence it remains to prove that our dynamic program
	embeds $T_2$ according to~$\mathcal{E}_2$,
	which we do in \cref{cl:t2planar}.
	In \cref{cl:mincrossings}, we show
	that the resulting drawing has the minimum number of crossings.
	This proves the correct behavior of our algorithm.
	In \cref{cl:runtime}, we also show the
	runtime bound of $\oh(n^3)$.
\end{proof}

\paragraph{Description of the Dynamic Program.}
Consider some layer $j \in \{2, \dots, \ell\}$
and index the vertices in $V_j(T_1)$ according $\mathcal{E}_1$
from left to right by $1, \dots, n_{1|j}$.
In a complete drawing, we define, for a vertex~$v$ of~$T_2$,
its \emph{position}~$p$ on layer~$j$ with respect to the 
index of the closest vertex of~$T_1$ to the left of~$v$.
If there is no such vertex, we set $p = 0$.
Let $C_v = \{ c_1, c_2, \dots, c_{\mathrm{indeg}(v)} \}$ be
the ordered set of children of~$v$ in $T_2$, which lie on layer~$j-1$,
where $\mathrm{indeg}(v)$ is the in-degree of~$v$.

For our dynamic program, we define a function~$o$,
which has as first parameter a vertex~$v$ of~$T_2$
and as second parameter a position~$p$ on layer~$L(v)$.
The value of~$o$ shall describe the number of crossings
in an optimal partial solution for the drawing of the subtree of~$T_2$
rooted at~$v$ and placed at position~$p$.
As usual in a dynamic program,
we compute a function value once and then save it in
a lookup table.
Additionally, we save the recursive dependencies that led to a value
to reconstruct a drawing in the end.
If $j \ge 3$, we define~$o$ as follows.

\begin{equation}
o[v,p] = \sum\limits_{i = 1}^{|C_v|} ~ \min\limits_{q \in \{0, \dots, n_{1|j-1}\} } \left( o[c_i,q] + \cro_{j-1}(q,p) \right) \nonumber
\end{equation}
where $\cro_{x}(y,z)$ is a
\emph{crossing function} describing
the number of crossings an edge between layers $x$
and \mbox{$x + 1$} admits if its source is arranged
at position $y$ (of layer~$x$) and its target is arranged
at position~$z$ (of layer~$x+1$).
If for some $c_i$, there is more than one position for $q$
resulting in a minimum value of~$o[v,p]$, we choose the position $q$ that maximizes $\cro_{j-1}(q,p)$.

For the recursive function~$o$, we add a terminating formulation
for the vertices on layer~$j = 2$.
Recall that for the leaves on layer~$1$,
there is a total order $<_1$ given.
Hence, for a vertex $v \in V_2(T_2)$, the position
of each child of~$v$ is fixed,
leading to the following simplified formulation of~$o$,
where $p_{c_i}$ is the given position of leaf $c_i$.
\begin{equation}
	o[v,p] = \sum\limits_{i = 1}^{|C_v|} \cro_1(p_{c_i},p) \nonumber
\end{equation}

To compute the value~$o^\star$ of the dynamic program as a whole,
we take the minimum of all values of~$o$ for the
root~$r_2$ of~$T_2$:
\begin{equation}
o^\star = \min_{p \in \{0, \dots, n_{1|L(r_2)}\}} o[r_2, p]. \nonumber
\end{equation}
We return a drawing corresponding to~$o^\star$, i.e.,
we specify for each vertex~$v$ of~$T_2$
its position with respect to~$T_1$ when computing~$o^\star$.
Finally, for vertices of~$T_2$ having the same position,
we arrange them in the order given by~$\mathcal{E}_2$.

\paragraph{Correctness.}
Next, we prove the correct behavior of our dynamic program
by showing that $T_2$ is embedded according to~$\mathcal{E}_2$
(\cref{cl:t2planar}) and by showing that the resulting drawing
has the minimum number of crossings (\cref{cl:mincrossings}).
Mainly because \cref{cl:t2planar} is rather intricate to prove,
we next introduce some more notation and concepts,
for which we show four claims
that lead to the proofs of these lemmas.

A key observation is that for a position~$p$ on layer~$j$,
there is precisely one \emph{ideal position}~$p^\star$ on layer~$j-1$
such that $\cro_{j-1}(p^\star,p) = 0$
and for two positions $p,q$ with $p < q$ on layer~$j$,
the ideal positions $p^\star, q^\star$ on layer~$j-1$ appear
in the same order, i.e., $p^\star < q^\star$.
(Imagine going down the gap of $\mathcal{E}_1$ where $p$ is located as illustrated in \cref{fig:sim-embedding-trees-b}.)
In \cref{cl:distancetoidealgap}, we formalize another
observation regarding ideal positions.
Essentially, the claim says that the further the endpoints
of an edge are away from a pair of position and ideal position,
the more crossings occur.
For simplicity, we assume henceforth that
each of the functions $o$ and $\cro_j$
returns $\infty$ for parameters outside of its domain.

\begin{myclaim}
	\label{cl:distancetoidealgap}
	On a layer~$j \in \{2, \dots, L(r_1)\}$,
	let $p \in \{0, \dots, n_{1|j}\}$ be a position
	and let $p^\star \in \{0, \dots, n_{1|j-1}\}$ be the ideal
	position of $p$ on layer~$j-1$.
	For any $x \in \mathbb{N}_0$, it holds that
	$\cro_{j-1}(p^\star \pm x,p) = x$ and
	$\cro_{j-1}(p^\star, p \pm (x+1)) > \cro_{j-1}(p^\star, p \pm x)) \ge x$.
\end{myclaim}
\begin{proof}
	Consider an edge $(u, v)$ of $T_2$
	with its endpoints being placed at $p^\star$ and~$p$.
	We know that $\cro_{j-1}(p^\star, p) = 0$.
	Now, for every position that we move $u$ ($v$, resp.) to the left or right
	of~$p^\star$ (of~$p$), we change sides with a vertex~$w$ of $T_1$.
	Because $w$ has exactly (at least) one incident
	edge going upwards (downwards) to its parent (a child)
	that we have not crossed before,
	the number of crossings increases by exactly (at least)~1.
\end{proof}

\begin{figure}
	\centering
	\begin{minipage}[b]{0.478 \linewidth}
		{\includegraphics[page=1]{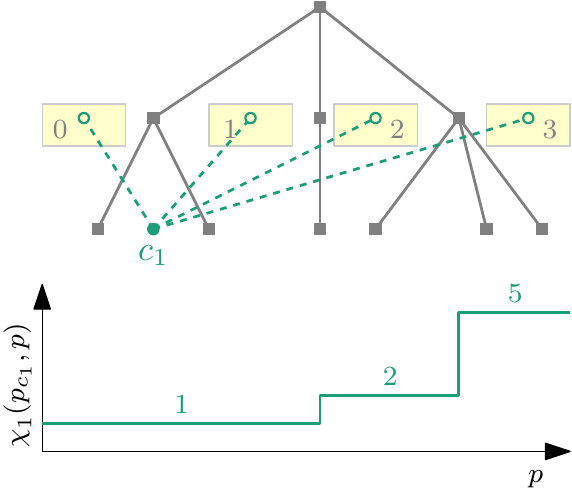} \raggedright}
		
		\smallskip
		
		{\footnotesize \begin{spacing}{1}
				\textsf{\bfseries (a)} Crossings of
				the edge $(c_1, v)$ dependent on
				the position $p$ of $v$.\end{spacing}}
		
		\bigskip
		
		{\includegraphics[page=3]{crossing-functions} \raggedright}
		
		\smallskip
		
		{\footnotesize \begin{spacing}{1}
				\textsf{\bfseries (c)} Crossings of
				the edge $(c_3, v)$ dependent on
				the position $p$ of $v$.\end{spacing}}
	\end{minipage}
	\hfill
	\begin{minipage}[b]{0.478 \linewidth}
		{\includegraphics[page=2]{crossing-functions} \raggedright}
		
		\smallskip
		
		{\footnotesize \begin{spacing}{1}
				\textsf{\bfseries (b)} Crossings of
				the edge $(c_2, v)$ dependent on
				the position $p$ of $v$.\end{spacing}}
		
		\bigskip
		\bigskip
		\bigskip
		\medskip
		\smallskip
				
		{\includegraphics[page=4]{crossing-functions} \raggedright}
		
		\smallskip
		
		{\footnotesize \begin{spacing}{1}
				\textsf{\bfseries (d)} Sum of  the three
				crossing functions gives $o[v,p]$.\end{spacing}}
	\end{minipage}
	
	\caption{Example of a vertex $v$ of $V_2(T_2)$ having three children $c_1, c_2, c_3$,
			where the position $p_c$ of a child $c$ and the position $p$ of $v$
			determine the value of $o[v,p]$.
			Here, we perceive $\cro$ and $o$ as functions dependent on $p$.}
	\label{fig:crossing-functions}
\end{figure}

For a vertex $v \in V_j(T_2)$ on a layer~$j$,
we define $\optpos(v)$ as the set of
every position~$p$ where $o[v,p]$ is minimum.
We analyze the properties of $\optpos$ in \cref{cl:minvaluessequential}.
\newpage
\begin{myclaim}
	\label{cl:minvaluessequential}
	For every layer $j \in \{2, \dots, L(r_2)\}$,
	let $v_1, v_2, \dots, v_{n_{2|j}}$
	be the vertices in $V_j(T_2)$ in the order of $\mathcal{E}_2$.
	It holds that
	$\min \optpos(v_1) \le
	\min \optpos(v_2) \le \ldots \le
	\min \optpos(v_{n_{2|j}})$ and
	$\max \optpos(v_1) \le
	\ldots \le
	\max \optpos(v_{n_{2|j}})$.
	
	Further, for every $v \in V_j(T_2)$,
	$\optpos(v)$ is an interval of
	natural numbers and, for any $x \in \mathbb{N}_0$,
	$o[v, \min \optpos(v) - (x+1)] > o[v, \min \optpos(v) - x]	\ge x$ and
	$o[v, \max \optpos(v) + (x+1)] > o[v, \max \optpos(v) + x] \ge x$.
\end{myclaim}
\begin{proof}
	We show this claim by induction over the layers $j = 2, 3, \dots$
	
	For $j = 2$ and every $v \in V_j(T_2)$,
	the children of $v$ have fixed positions and, therefore,
	$o[v,p]$ only depends on the number of crossings
	induced by the position~$p \in \{0, \dots, n_{1|j}\}$;
	see \cref{fig:crossing-functions} for an example.
	We next show that $\optpos(v)$ is an interval.
	Observe that $o[v,p] = \sum_{i = 1}^{|C_v|} \cro_1(p_{c_i},p)$
	is a sum of discrete functions (with variable~$p$)
	where each admits its minimum value
	for one or two neighboring values of~$p$
	and apart from at most two values at or around this minimum,
	all of these functions increase or decrease
	by the same amount if we add or subtract~1 to~$p$,
	which follows by the argument presented in the proof
	of \cref{cl:distancetoidealgap}.
	(These functions here are weakly unimodal, i.e.,
	they have a global minimum and they increase
	monotonously when moving away from that minimum.)
	Now to find the positions where that sum is minimum,
	we traverse the values of its domain:
	we start with $p = 0$.
	If we increase $p$ by one,
	then all functions that have not yet reached
	its minimum decrease, while the functions that
	had already reached their minimum increase
	by the same amount.
	Hence, this sum is minimum in the interval
	of the domain that has the minima of
	the single crossing functions equally distributed
	on the left and on the right side.
	Furthermore, for each position further to the left or right,
	the sum increases by at least one.
	It remains to show that the minima and maxima
	of $\optpos(v_1), \optpos(v_2), \dots, \optpos(v_{n_{2|j}})$
	increase monotonously.
	Since the children of the vertices on layer~$2$
	(i.e., the leaves) are ordered,
	the minima and maxima of all crossing functions are ordered
	and so are the minima and maxima of the sums.
	
	Now consider $j > 2$.
	Again $o[v,p]$ is a sum,
	but now we add, for every child~$c$ of $v$
	and a position $q$, $o[c,q]$ and $\cro_{j-1}(q,p)$.
	The sum of minima is again a sum of unimodal
	functions with similar properties as before:
	the $\cro_{j-1}(q,p)$ summands increase and decrease
	around their minimum as before,
	while the $o[c,q]$ summands increase and decrease
	before and after their minimum at least
	as much as a $\cro_{j-1}(q,p)$ summand
	due to the induction hypothesis.
	(They behave like a weighted $\cro_{j-1}(q,p)$ summand.)
	Hence, if we sum them up,
	we apply a weighted version of the previous argument
	to obtain the properties stated in the claim.
	In particular, the minima and maxima
	of $\optpos(v_1), \optpos(v_2), \dots, \optpos(v_{n_{2|j}})$
	increase monotonously since the
	minima and maxima of $\optpos$ of the children on layer~$j-1$
	are ordered by the induction hypothesis.
\end{proof}

For a vertex $u$ on a layer~$j-1$,
we define the \emph{natural position} $\natpos(u, p)$ of $u$
with respect to the position $p$ of its parent vertex
on layer~$j$ as
\begin{equation*}
	\natpos(u,p) =
	\begin{cases}
		p^\star, & \textrm{ if } p^\star \in \optpos(u) \\
		\max \optpos(u), & \textrm{ if } p^\star > \max \optpos(u) \\
		\min \optpos(u), & \textrm{ if } p^\star < \min \optpos(u).
	\end{cases}
\end{equation*}
In \cref{cl:ovaluesmin}, we describe, for a vertex $v$,
the behavior of the natural positions of $v$'s
children and their relationship to $o[v,p]$.

\begin{myclaim}
	\label{cl:ovaluesmin}
	For a vertex $v \in V_j(T_2)$ on a layer~$j$,
	let $c_1, \dots, c_{|C_v|}$ be the children of~$v$.
	For any position~$p \in \{0, \dots, n_{1|j}\}$,
	it holds that $\natpos(c_1, p) \le \ldots \le \natpos(c_{|C_v|}, p)$ and
	$o[v,p] = \sum_{i = 1}^{|C_v|} (o[c_i, \natpos(c_i, p)] + \cro_{j-1}(\natpos(c_i, p),p))$.
\end{myclaim}
\begin{proof}
	We partition the children of~$v$ into three groups:
	if for a child $c_i$ (where $i \in \{1, \dots, |C_v|\}$),
	$p^\star \in \optpos(c_i)$, we set $q_i = p^\star$.
	If for a child $c_i$, $p^\star > \max \optpos(c_i)$,
	we set $q_i = \max \optpos(c_i)$,
	and, symmetrically, if $p^\star < \min \optpos(c_i)$,
	we set $q_i = \min \optpos(c_i)$.
	By \cref{cl:minvaluessequential},
	we observe that $q_1 \le \ldots \le q_{|C_v|}$.
	Since $q_i = \natpos(c_i, p)$, this proves the first part of the claim.
	
	Now for the second part,
	if $o[v,p] \ne \sum_{i = 1}^{|C_v|} (o[c_i, q_i] + \cro_{j-1}(q_i,p))$,
	then for some~$i$,
	$o[c_i, q_i] + \cro_{j-1}(q_i,p) > \min_{q' \in \{0, \dots, n_{1|j-1}\}} (o[c_i, q'] + \cro_{j-1}(q',p))$.
	Let $\hat{q} \in \{0, \dots, n_{1|j-1}\}$
	be a position such that
	$o[c_i, q_i] + \cro_{j-1}(q_i,p) > o[c_i, \hat{q}] + \cro_{j-1}(\hat{q},p)$.
	Since we know $o[c_i, q_i] \le o[c_i, \hat{q}]$,
	it follows that $\cro_{j-1}(q_i,p) > \cro_{j-1}(\hat{q},p)$.
	
	First note that $p^\star \notin \optpos(c_i)$
	because otherwise $\cro_{j-1}(q_i,p) < \cro_{j-1}(\hat{q},p)$.
	It follows that $q_i$ is the minimum or maximum position
	of $\optpos(c_i)$~-- assume w.l.o.g.\ that $q_i = \max \optpos(c_i)$.
	Then, $q_i < \hat{q}$ (and hence $\hat{q} \notin \optpos(c_i)$)
	because otherwise again $\cro_{j-1}(q_i,p) < \cro_{j-1}(\hat{q},p)$.
	
	We distinguish two cases.
	The first case is $q_i < \hat{q} \le p^\star$. 
	By \cref{cl:distancetoidealgap},
	we know that $\cro_{j-1}(q_i,p) - \cro_{j-1}(\hat{q},p) = \hat{q} - q_i$.
	By \cref{cl:minvaluessequential},
	we know that $o[c_i, \hat{q}] - o[c_i, q_i] \ge \hat{q} - q_i$.
	Hence, we have $o[c_i, \hat{q}] - o[c_i, q_i] \ge \cro_{j-1}(q_i,p) - \cro_{j-1}(\hat{q},p)$,
	which we can reformulate as
	$o[c_i, q_i] + \cro_{j-1}(q_i,p) \le o[c_i, \hat{q}] + \cro_{j-1}(\hat{q},p)$,
	which contradicts our initial assumption.
	
	The second case is $q_i < p^\star < \hat{q}$.
	Now we have
	$\cro_{j-1}(q_i,p) = p^\star - q_i$
	and $\cro_{j-1}(\hat{q},p) = \hat{q} - p^\star$.
	If we add up these two equations, we get
	$\cro_{j-1}(q_i,p) + \cro_{j-1}(\hat{q},p) = \hat{q} - q_i$.
	As we still have $o[c_i, \hat{q}] - o[c_i, q_i] \ge \hat{q} - q_i$,
 	we get $o[c_i, q_i] + \cro_{j-1}(q_i,p) \le o[c_i, \hat{q}] - \cro_{j-1}(\hat{q},p)$,
 	which, of course, also contradicts our initial assumption.
\end{proof}

Now in the last claim, which is \cref{cl:dp-positions},
we directly investigate the positions
that are chosen by our dynamic program
as the positions of the children of a vertex~--
they turn out to be the natural positions.

\begin{myclaim}
	\label{cl:dp-positions}
	For a vertex $v \in V_j(T_2)$ on a layer~$j$
	and a position~$p \in \{0, \dots, n_{1|j}\}$,
	the dynamic program selects
	$\natpos(c_1, p), \dots, \natpos(c_{|C_v|}, p)$
	as the positions of $v$'s children
	$c_1, \dots, c_{|C_v|}$ on layer~$j-1$.
\end{myclaim}
\begin{proof}
	Recall that, for any $i \in \{1, \dots, |C_v|\}$,
	if there is more than one position for $c_i$ resulting in a minimum value of~$o[v,p]$,
	the position of $q$ with the maximum value of $\cro_{j-1}(q,p)$ is used as a tie-breaker rule.
	If $\natpos(c_i, p) = p^\star$,
	then $p^\star$ is the only position of $c_i$
	that can lead to a minimum value of $o[v,p]$
	and our claim is true.
	
	Now, due to symmetry,
	we assume w.l.o.g.\ that $\natpos(c_i, p) = \max \optpos(c_i)$.
	Let $p' \ne \natpos(c_i, p)$ be a position of $c_i$
	yielding a minimum value of $o[v,p]$.
	The position $p'$ cannot lie within $\optpos(c_i)$ by~\cref{cl:distancetoidealgap} since this would result in a larger number of crossings, while $o[c_i, p'] = o[c_i, \natpos(c_i, p)]$.
	Hence, $\natpos(c_i, p) < p'$.
	By \cref{cl:minvaluessequential},
	$o[c_i, p'] - o[c_i, \natpos(c_i, p)] \ge p' - \natpos(c_i, p)$.
	This means, that, for each position
	further to the right of $\natpos(c_i, p)$, the value
	of the dynamic program for $c_i$ increases by at least one,
	while the number of crossings according to the function
	$\cro_{j-1}$ increases by exactly one
	(see \cref{cl:distancetoidealgap}).
	Thus, $\natpos(c_i, p)$ is one
	(of possibly several) position(s) of $c_i$
	admitting a minimum value of $o[v,p]$.
	If $p'$ also admits a minimum value of $o[v,p]$,
	but $o[c_i, p'] > o[c_i, \natpos(c_i, p)]$,
	it follows that $\cro_{j-1}(p', p) < \cro_{j-1}(\natpos(c_i, p), p)$.
	Hence, due to the tie-breaker rule,
	our algorithm would have selected $\natpos(c_i, p)$ instead of~$p'$.
\end{proof}

Now we have gathered everything to establish the key lemma of this section.

\begin{lemma}
	\label{cl:t2planar}
	The drawing of $T_2$ is embedded according to~$\mathcal{E}_2$.
\end{lemma}
\begin{proof}
We prove that the vertices of $T_2$
are ordered according to~$<^2_j$
by induction on layer~$j$, starting with the layer of the root $r$ of $T_2$.
On layer $j = L(r_2)$, there is only $r$,
which, of course, cannot contradict $<^2_j$.

Let $j < L(r_2)$ and $v_1, v_2, \dots v_{n_{2|j+1}}$ be the vertices on layer~$j+1$.
Then, by \cref{cl:ovaluesmin,cl:dp-positions},
the children $C_{v_i}$ on layer~$j$ have increasing positions
respecting~$<^2_j$ for every $i \in \{1, \dots, n_{2|j+1}\}$
and the edges to these children do not cross.
It remains to show that no pair of edges between layers~$j$ and $j+1$ without a common endpoint cross.
By our induction hypothesis,
for two vertices $v_i,v_{i'}$ on layer~$j+1$ with $i < i'$,
the position $p_i$ of $v_i$ is not greater
than the position $p_{i'}$ of $v_{i'}$.
By \cref{cl:distancetoidealgap}, it follows for the ideal positions of $p_i$ and $p_{i'}$ that $p^\star_i \le p^\star_{i'}$.
By \cref{cl:minvaluessequential}, the $\min$ and $\max$ values of $\optpos$ of the vertices on layer~$j$ are monotonically increasing.
Hence, by \cref{cl:dp-positions} and the definition
of $\natpos$, there are no crossings between
edges with target $v_i$ and edges with target $v_{i'}$,
as this would contradict $p^\star_i \le p^\star_{i'}$.
Hence, the edges between layer~$j$ and $j+1$ are planar, which concludes the induction step.
\end{proof}
After we have now shown that the dynamic program
yields a valid solution, i.e.,
a drawing where both trees are internally crossing-free,
it remains to prove that the number of crossings
between $T_1$ and~$T_2$ is minimum.
\begin{lemma}
	\label{cl:mincrossings}
	The number of crossings in the computed drawing is minimum.
\end{lemma}
\begin{proof}
	We show by induction over the layers $j = 2, 3, \dots$
	that for a vertex $v$ and a position $p$,
	$o[v, p]$ is the minimum number of crossings
	induced by ($T_1$ and) the subtree of $T_2$ rooted at~$v$,
	which we call~$T_v$,
	across all drawings of $T_v$ when we place~$v$ at position~$p$.
	For $j = 2$, this is clear as we just sum up the
	number of crossings induced by the edges to the leaves.
	
	Let $j \ge 3$.
	By \cref{cl:dp-positions},
	we know that the dynamic program has selected the positions
	$\natpos(c_1, p) \le \ldots \le \natpos(c_{|C_v|}, p)$
	for the children $c_1, \dots, c_{|C_v|}$ of~$v$.
	By our induction hypothesis, we know that, for each $c \in C_v$,
	$o[c, \natpos(c, p)]$ corresponds to a drawing
	of $T_c$ at position $\natpos(c, p)$ with the minimum number of crossings.
	We add up the number of crossings between layer~$j-1$ and $j$,
	and by the formulation of the dynamic program, we know
	that this is again minimum across all positions of~$c$.
\end{proof}

\paragraph{Running Time.}
It remains to analyze the running time
of our dynamic program.

\begin{lemma}
	\label{cl:runtime}
	The running time of our algorithm is in $\oh(n_1^2 \cdot n_2) \subseteq \oh(n^3)$.
\end{lemma}
\begin{proof}
	For a vertex~$v$ of~$V_j(T_2)$ and a position $p \in \{0, \dots, n_{1|j}\}$,
	we can compute $o[v, p]$ by finding,
	for each child $c \in C_v$ and
	a position $q$ in a subset of $\{0, \dots, n_{1|j-1}\}$,
	the minimum of $o[c,q] + \cro_{j-1}(q,p)$.
	
	The number of children over all steps is in~$\oh(n_2)$ as~$T_2$ is a tree
	and the number of positions is in~$\oh(n_1)$.
	We can pre-compute and store all values~$\cro_j(q,p)$
	in $\oh(n_1^2)$ time.
	We have $\oh(n_1 n_2)$ entries of $o[v, p]$,
	which we can compute in overall $\oh(n_1^2 n_2) \subseteq \oh(n^3)$ time.
	The optimal root placement can be found in linear time.
	For the backtracking when constructing the final drawing,
	we simply store for each entry $o[v, p]$ a pointer to the entries it is based on.
\end{proof}

\section{Multiple Trees on Three Layers}\label{sec:three-layers}
\newcounter{secThreeLayers}
\setcounter{secThreeLayers}{\value{section}}

In this section, we consider the case that we are given
a forest~$\mathcal{F} = \{T_1, \dots, T_k\}$
of $k$ trees spanning (at most) three layers each,
and we show the following result.

\begin{theorem}
	\label{thm:three-layers}
	Let $\mathcal{F}$ be an $n$-vertex layered forest of $k$ rooted trees
	on three layers,
	where all leaves are assigned to layer~$1$ and have a fixed order,
	which prescribes a planar embedding of each tree individually.
	We can compute a drawing of~$\mathcal{F}$
	where each tree is drawn in the prescribed planar embedding
	with the minimum number of crossings in $\oh(n^k)$~time.
\end{theorem}

The first property we use to prove \cref{thm:three-layers}
is that the order of roots on layer~$3$ is
fixed, similar to the order of the leaves on layer~$1$.
We can assume this because there are only up to $k$ roots on layer~$3$, with at most $k!$ ways to arrange them. 
We simply consider each permutation of the roots on layer~$3$ individually,
and henceforth assume that both total orders $<_1$ and $<_3$ are given, fixing the roots and leaves, and the only remaining task is to compute $<_2$ of the vertices on layer~$2$ while maintaining their partial order~$\prec_2$.
Note that if any tree has its root on layer~$2$,
we treat this root like the other vertices of layer~$2$.

As in \cref{sec:two-trees}, we use the notion of
positions and crossing functions, however,
we slightly adjust their definitions to better suit
the setting of this section.
Let~$\sigma$ be a permutation of~$V_2(\mathcal{F})$
indexed by $1, 2, \dots$ and respecting the partial order $\prec_2$.
For $i \in \{1, \dots, k\}$ and some vertex $v \in V_2(\mathcal{F}) \setminus V_2(T_i)$,
we denote the \emph{position} (starting at~0) of~$v$ within the subsequence of~$\sigma$
consisting of the vertices~$V_2(T_i) \cup \{v\}$ by~$p_i^v$.%
\footnote{This is a generalization of the positions
	introduced in \cref{sec:two-trees} where all
	positions were relative to (the given embedding of) $T_1$.}
Note that, in a drawing using~$\sigma$ as~$<_2$, we can charge every crossing to precisely two vertices of~$V_2(\mathcal{F})$ as
any crossing occurs between two edges that have two distinct endpoints on layer~$2$.
Now observe that for a vertex~$v \in V_2(T_j)$, where $j \in \{1, \dots, k\}$,
the number of crossings charged to~$v$ with respect to~$\sigma$
depends only on $p_i^v$ for each $i \in \{1, \dots, k\} \setminus \{j\}$.
Therefore, we introduce the \emph{crossing function}~$\cro^v_i(p)$
returning the resulting number of crossings
when we insert $v$ at a position $p \in \{0, \dots, n_{i|2}\}$
into the planar embedding of~$T_i$.
The number~$\cro_\sigma(v)$ of crossings charged to~$v$
when using permutation~$\sigma$ is then
$\cro_\sigma(v) = \sum_{i \in \{1, \dots, k\} \setminus \{j\}} \cro^v_i(p_i^v)$
and the total number~$\cro(\sigma)$ of crossings when using
permutation~$\sigma$ is then
$\cro(\sigma) = \sum_{v \in V_2(\mathcal{F})} \cro_\sigma(v) / 2$.

\begin{lemma}
	\label{lem:compute-civp}
	For all combinations of $i \in \{1, \dots, k\}$,
	$v \in V_2(\mathcal{F}) \setminus V_2(T_i)$,
	and $p \in \{0, \dots, n_{i|2}\}$,
	we can compute every value~$\cro^v_i(p)$ in a total of $\oh(n^2)$ time.
\end{lemma}
\begin{proof}
	First save, for every $v \in V_2(\mathcal{F})$,
	the star~$S_v$ induced by $v$ and $v$'s neighbors
	in total $\oh(n)$ time.
	Now for a fixed $i \in \{1, \dots, k\}$, consider the given
	planar embedding~$\mathcal{E}_i$ of~$T_i$.
	Also fix $v \in V_2(\mathcal{F}) \setminus V_2(T_i)$ and
	compute~$\cro^v_i(0)$ by checking,
	for every pair of edges of $S_v$ and $T_i$,
	if there is a crossing if $v$ is the leftmost vertex on layer~$2$.
	Then for $p = 1, \dots, n_{i|2}$,
	update $\cro^v_i(p-1)$ to $\cro^v_i(p)$ by checking each pair
	of edges from $S_v$ and the star around the $p$-th vertex of~$V_2(T_i)$.
	Over all of these steps, all vertices, and all trees,
	every pair of edges is considered at most four times,
	which yields a running time in~$\oh(n^2)$.
\end{proof}

\begin{figure}[t]
	\centering
	\begin{minipage}[b]{0.485 \linewidth}
		{\includegraphics[page=1,scale=.945]{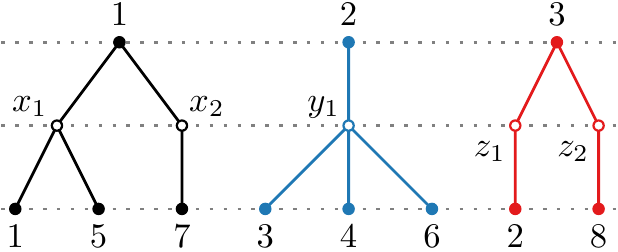} \raggedright}
		
		\smallskip
		
		{\footnotesize \begin{spacing}{1}
				\textsf{\bfseries (a)} We are given
				three embedded trees on three layers
				together with a total order of
				the leaves and the roots
				(numbers on the top and the bottom side).
			\end{spacing}}
		
		\bigskip
		
		{\includegraphics[page=2,scale=.945]{three-layers-example} \raggedright}
		
		\smallskip
		
		{\footnotesize \begin{spacing}{1}
				\textsf{\bfseries (b)} Drawing of the
				trees from (a) with the minimum number of crossings
				(six pairwise crossings) where
				the orders of the leaves and the roots
				are given.
				The total order of the vertices
				on layer~$2$, the middle layer, corresponds to the
				shortest $st$-path highlighted in orange
				in (c).
		\end{spacing}}
	\end{minipage}
	\hfill
	\begin{minipage}[b]{0.475 \linewidth}
		{\raggedleft \includegraphics[page=3,scale=.95]{three-layers-example}}
		
		\smallskip
		
		{\footnotesize \begin{spacing}{1}
				\textsf{\bfseries (c)} Grid graph~$H$
				with edge weights whose $st$-paths
				represent precisely the (allowed) total orders
				of the vertices on layer~$2$
				of the forest shown in~(a).
				The width in x-dimension is two and represents
				first choosing vertex $x_1$ and then vertex $x_2$
				of the first tree.
				Similarly, the y- and z-dimension represent the
				vertices of the second and third tree, respectively.
				The $st$-path highlighted in orange is the lightest path with weight~$12$.
			\end{spacing}}
	\end{minipage}
	
	\caption{Reducing the problem of finding a layered drawing of $k$~trees
		on three layers with the minimum number of crossings,
		where the leaves and the roots are fixed,
		to a shortest-path problem in a weighted $k$-dimensional grid graph.}
	\label{fig:three-layers-example}
\end{figure}

\paragraph{Reduction to a Shortest-Path Problem.}
We now construct a weighted directed acyclic $st$-graph~$H$
(see \cref{fig:three-layers-example}c)
whose $st$-paths represent precisely all total orders of $V_2(\mathcal{F})$
that respect the vertex orders $<_2^1, \dots, <_2^k$
given for each tree by its prescribed planar embedding
(see \cref{fig:three-layers-example}a).
Moreover, for an $st$-path~$\pi$ representing
a total order~$\sigma$ of $V_2(\mathcal{F})$,
the weight of $\pi$ is twice the number of crossings induced by~$\sigma$.
We let $H$ be the $k$-dimensional grid graph of side lengths $n_{1|2} \times \dots \times n_{k|2}$ directed from one corner to an opposite corner.
More precisely, $H$ has the node set $\{ (x_1, \dots, x_k) \mid x_1 \in \{0, \dots, n_{1|2}\}, \dots, x_k \in \{0, \dots, n_{k|2}\} \}$
and there is a directed edge from $(x_1, \dots, x_k)$
to $(y_1, \dots, y_k)$ if $x_j + 1 = y_j$ for exactly one $j \in \{0, \dots, k\}$ and $x_i = y_i$ otherwise.
Observe that, within~$H$, $(0, \dots, 0)$ is the unique source
and $(n_{1|2}, \dots, n_{k|2})$ is the unique sink,
which we denote by~$s$ and~$t$, respectively.
We let an edge~$e$ from $(x_1, \dots, x_k)$ to $(y_1, \dots, y_k)$ where $x_j + 1 = y_j$
represent (i) taking the $y_j$-th vertex of~$V_2(T_j)$,
to which we refer as~$v$ next,
(ii)~after having taken $x_i$ vertices of $V_2(T_i)$ for each $i \in \{1, \dots, k \}$.
Thus, we let the weight~$w_e$ of $e$ in $H$ be the number of crossings charged to~$v$
in this situation, that is, $w_e = \sum_{i \in \{1, \dots, k\} \setminus \{j\}} \cro^v_i(x_i)$.

Clearly, any $st$-path~$\pi$ in $H$ has (unweighted) length~$n_2$.
If we traverse $\pi$,
we can think of layer~$2$ as being empty when we start at $s$,
and then, for each edge of~$\pi$,
we take the corresponding vertex of~$V_2$ and add it to layer~$2$.
Since edge weights equal the number of crossings
the corresponding vertices would induce in this situation,
finding a lightest $st$-path in $H$ means
finding a crossing-minimal total order of layer~$2$
(see \cref{fig:three-layers-example}).
By constructing $H$ (using \cref{lem:compute-civp}
to compute the edge weights)
and searching for an $st$-path of minimum weight,
we obtain an XP-algorithm in $k$;
see \cref{thm:three-layers},
which we formally prove next.

\begin{proof}[of \cref{thm:three-layers}]
	For $\mathcal{F}$, we fix each order of layer~$3$ once
	and compute the corresponding
	$k$-dimensional grid graph~$H$.
	We first argue that the $st$-paths of~$H$
	represent precisely the possible total orders of $V_2(\mathcal{F})$
	and their weights are twice the number of crossings
	in the corresponding drawing of~$\mathcal{F}$.
	Thereafter, we argue about the running time.
	
	Since we have a $k$-dimensional grid graph,
	any $st$-path in~$H$ traverses
	$n_{1|2}$ edges in the first dimension,
	$n_{2|2}$ edges in the second dimension, etc.
	We can interleave the edges of different dimensions arbitrarily
	to obtain different paths.
	Hence, each path is equivalent to exactly one total order~$\sigma$
	extending the partial order $\prec_2$,
	which is given by $<_2^1, \dots, <_2^k$.
	The number of crossings of a drawing only depends on~$\sigma$.
	Every crossing occurs between two edges being incident
	to precisely two distinct vertices of $V_2(\mathcal{F})$.
	We \emph{charge} the crossing to these two vertices.
	Hence, we can add up the numbers of crossings charged
	to each vertex~$v$ (i.e., $\cro_\sigma(v)$) and divide the sum by two.
	These numbers of crossings charged to the vertices are
	by definition the edge weights of~$H$.
	Therefore, each minimum-weight $st$-path in~$H$
	is equivalent to a minimum-crossing drawing
	of $\mathcal{F}$ for the given order of leaves and roots.
	
	It remains to argue about the running time.
	The number of directed edges of~$H$ is upper-bounded by
	\begin{align*}
	E(H) &= \sum\limits_{j = 1}^k n_{j|2} \prod\limits_{i \in \{1, \dots, k\} \setminus \{j\}} \left( n_{i|2} + 1 \right) \\
	&\le k \prod\limits_{i = 1}^k \left( n_{i|2} + 1 \right)
	\le k \left(\frac nk\right)^k.
	\end{align*}
	To compute the weight of each such edge,
	we sum up $k - 1$ values of $\cro^v_i(p)$,
	which we have pre-computed in $\oh(n^2)$ time using \cref{lem:compute-civp}.
	Therefore, we can construct $H$ including
	the assignment of edge weights in $\oh(k^2 (n/k)^k)$ time
	and we can find a minimum-weight path in~$H$ in $\oh(k (n/k)^k)$ time
	using topological sorting.
	Recall that we construct a graph~$H$ for at most $k!$ permutations
	of the roots on layer~$3$.
	For the final minimum-crossing drawing,
	we use the permutation of $V_2(\mathcal{F})$
	and the permutation of $V_3(\mathcal{F})$
	that correspond to the lightest minimum-weight path in any~$H$.
	Hence, the total running time is in $\oh(k! k^2(n/k)^k) \subseteq \oh(k^3 \cdot (k-1) \cdot \ldots \cdot 2 \cdot 1 / k^k \cdot n^k) \subseteq \oh(n^k)$.
\end{proof}

Finally, we remark that our XP-algorithm from \cref{thm:three-layers} can be generalized in two ways.

\begin{remark}
	By definition, $\prec_2$ has only constraints between vertices
	of the same tree.
	We can extend $\prec_2$ by (arbitrarily many) constraints between
	vertices of different trees and \cref{thm:three-layers} still holds.
	This is because we can easily adjust our reduction:
	say $x$ is the $i$-th vertex on layer~$2$ of the first tree,
	$y$ is the $j$-th vertex on layer~$2$ of the second tree,
	and let the constraint $x \prec_2 y$ be given.
	Then, in~$H$, we set the weight of every edge representing~$x$
	and lying in the y-dimension at a position $\ge j$ to~$\infty$.
	Symmetrically, we set the weight of every edge representing~$y$
	and lying in the x-dimension at a position $< i$ to~$\infty$.
	This way, we prevent that a lightest path chooses an
	edge representing~$y$ before it chooses an edge representing~$x$.
	If and only if there is an $st$-path with non-infinity weight in~$H$,
	there is a valid arrangement of the vertices on layer~$2$.
\end{remark}

\begin{remark}
	Requiring trees on the three layers is a stronger restriction
	than actually needed.
	For our reduction, we only use the property that
	the vertex order on layer~$3$ is fixed,
	which we achieve by trying all permutations.
	For this approach, it suffices if layer~$3$ is sparse.
	Hence, our result also holds for $k$ planar-embedded graphs
	provided that on layer~$3$, there are $\oh(k)$ vertices.
	Moreover, for planar-embedded graphs and an arbitrary number of vertices on layer~$3$,
	\cref{thm:three-layers} holds as well
	if the total order of vertices on layer~$3$ is prescribed.
\end{remark}

\section{Conclusion and Open Problems}\label{sec:open}
In this work, we approach the problem of crossing minimization of layered rooted trees from two directions.
First, by describing a cubic-time dynamic program in \cref{thm:Two-trees}, keeping the number of trees $k$ small,
namely $k=2$, while allowing an arbitrary number of layers.
Inversely, our second result stated in \cref{thm:three-layers} is an XP-time algorithm for an arbitrary number of trees, restricted to only three layers. 
Hence, there is a gap between these two results, which has not yet been explored and naturally raises the following open problem.
Going one step further, is the case $k = 3$ trees and $\ell = 4$ layers polynomial-time solvable, and if so, for which $k$ and $\ell$ does it become hard?
Moreover, we pose the question of improving the complexity class for the case of $\ell = 3$ and $k > 2$, namely, can we solve the case of three layers in FPT-time in the number~$k$ of trees?
Alternatively this may be proved to be W[1]-hard.
Lastly, note that in our setting, we require that every tree
preserves its given planar embedding (imposed by the order of its leaves).
It is not clear, whether there exists a solution with less crossings without this restriction, although our current believe is that,
in any minimum-crossing solution, all of them are drawn planar.

\bigskip
\noindent \textbf{Acknowledgments.}
We thank the organizers of the workshop GNV 2022
in Heiligkreuztal for the fruitful atmosphere
where some of the ideas of this paper arose.
We also thank the anonymous reviewers for
their helpful feedback.

\bibliography{layered-trees}

\end{document}